\documentclass[10pt,reqno,oneside]{amsart}
\usepackage{verbatim,amsmath,amssymb,cite,epsfig,subfigure,color,ulem}

\numberwithin{equation}{section}

\newcommand{\Real}{\mathbb R}
\newcommand{\N}{\mathbb N}
\newcommand{\Z}{\mathbb Z}

\newcommand{\bfr}{{\bf r}}

\newcommand{\bfu}{{\bf u}}

\newcommand{\bfx}{{\bf x}}

\newcommand{\bfz}{{\bf z}}

\newcommand{\bsomega}{\boldsymbol{\omega}}
\newcommand{\bstheta}{\boldsymbol{\theta}}

\newtheorem{lemma}{Lemma}[section]

\newtheorem{remark}{Remark}[section]

\begin{document}

\title[]
{On the axisymmetric steady incompressible Beltrami flows}

\author{Pavel B\v{e}l\'{\i}k}
\author{Xueqing Su}
\author{Douglas P.~Dokken}
\author{Kurt Scholz}
\author{Mikhail M.~Shvartsman}

\address{P.~B\v{e}l\'{\i}k\\
Mathematics, Statistics, and Computer Science Department\\
Augsburg University\\
2211 Riverside Avenue\\
Minneapolis, MN 55454\\
U.S.A.}
\email{belik@augsburg.edu}

\address{X.~Su\\
Mathematics, Statistics, and Computer Science Department\\
Augsburg University\\
2211 Riverside Avenue\\
Minneapolis, MN 55454\\
U.S.A.}
\email{sux@augsburg.edu}

\address{D.~P.~Dokken\\
Department of Mathematics\\
University of St.~Thomas\\
2115 Summit Avenue\\
St.~Paul, MN 55105\\
U.S.A.} \email{dpdokken@stthomas.edu}

\address{K.~Scholz\\
Department of Mathematics\\
University of St.~Thomas\\
2115 Summit Avenue\\
St.~Paul, MN 55105\\
U.S.A.} \email{k9scholz@stthomas.edu}

\address{M.~M.~Shvartsman\\
Department of Mathematics\\
University of St.~Thomas\\
2115 Summit Avenue\\
St.~Paul, MN 55105\\
U.S.A.} \email{mmshvartsman@stthomas.edu}


\keywords{Axisymmetric Beltrami flow, Trkalian flow, Bragg--Hawthorne equation, cylindrical coordinates, spherical coordinates, paraboloidal coordinates, prolate spheroidal coordinates, oblate spheroidal coordinates, vorticity, vortex breakdown}

\subjclass[2010]{33C05,33C15,34B08,35Q31,35Q86,76U05,76B47}

\date{\today}

\begin{abstract}
  In this paper, Beltrami vector fields in several orthogonal coordinate systems are obtained analytically and numerically. Specifically, axisymmetric incompressible inviscid steady state Beltrami (Trkalian) fluid flows are obtained with the motivation to model flows that have been hypothesized to occur in tornadic flows. The studied coordinate systems include those that appear amenable to modeling such flows: the cylindrical, spherical, paraboloidal, and prolate and oblate spheroidal systems. The usual Euler equations are reformulated using the Bragg--Hawthorne equation for the stream function of the flow, which is solved analytically or numerically in each coordinate system under the assumption of separability of variables. Many of the obtained flows are visualized via contour plots of their stream functions in the $rz$-plane. Finally, the results are combined to provide a qualitative quasi-static model for a progression of flows that one might see in the process of a vortex breakdown. The results in this paper are equally applicable in electromagnetics, where the equivalent concept is that of a force-free magnetic field.
\end{abstract}

\allowdisplaybreaks
\thispagestyle{empty}
\maketitle

\section{Introduction}
\label{sec:intro}
In fluid mechanics, Beltrami or helical flows are fluid flows in which the velocity and the vorticity (curl of velocity) of the fluid are parallel to each other at all points and all times. Flows of this nature have been studied since at least the late 1800s and have applications in both fluid dynamics and electromagnetics, where a force-free magnetic field is one for which the Lorentz (magnetic) force density vanishes, or equivalently, the magnetic field is everywhere parallel to the direction of the current flow \cite{marsh96}. In this paper we will focus mainly on the hydrodynamics case, but many parallels can be drawn between the two, and results from one field can be applied in the other one.

Beltrami fluid flows are of interest for several reasons. They can have complex dynamics \cite{constantinmajda88}, and types of Beltrami flows that possess ergodic theoretic properties, such as strong mixing, make for attractive models for turbulent flows \cite{majdabertozzi01}. In nature, classes of rotating thunderstorms (supercells) may exhibit characteristics of Beltrami flows; this is supported by numerical simulations such as in \cite{nodaniino10}. Flows with low instability (CAPE) and nearly circular hodographs approach Beltrami flows as the hodograph becomes more circular~\cite{weismanrotunno00}. Highly helical flows are thought to be present in well-developed tornadic flows \cite{sasaki14} or resemble the gross supercell structure \cite{daviesjones08,lilly83,lilly86b}. Beltrami flows that are solutions of the Euler equations and the decaying nonsteady Beltrami flow solutions of the Navier--Stokes equations (e.g., \cite{daviesjonesrichardson02,lilly83,lilly86b}) could be or have been used to test code in numerical weather models, such as in the Advanced Regional Prediction System (ARPS) \cite{shapiro93}, the Weather Research and Forecasting Model (WRF) \cite{wrf}, or the Terminal Area Simulation System (TASS) \cite{ahmadproctor11}.

In this paper we investigate axisymmetric incompressible inviscid steady state Beltrami (Trkalian) flows with the goal of potentially developing other test cases for numerical models as well as possible models for tornadic or supercell flows. We explore several geometries characterized by various orthogonal coordinate systems and construct separable solutions to the relevant equations in these systems. We further explore symmetries in the mathematical problem and outline a way to construct infinitely many other Beltrami flows in each coordinate system. We present graphical results of such flows in each coordinate system.

The paper is organized as follows. In Section~\ref{sec:Beltrami} we review the governing equations for inviscid fluid flows and discuss Beltrami flows. In Section~\ref{sec:Bragg} we reformulate the axisymmetric problem using the Bragg--Hawthorne equation and discuss additional symmetries that can be used to construct new solutions from existing ones. In Section~\ref{sec:coordinates}, several orthogonal coordinate systems of interest are introduced, and under the assumption of separability of variables the mathematical problem is reformulated as a coupled system of two second-order ordinary differential equations, which are equipped with either zero, one, or two boundary conditions. In several cases, one of the equations results in a boundary value eigenvalue problem. These equations are then analyzed and in some cases analytic solutions are presented. Many of these solutions are given in terms of special functions (classical hypergeometric and Kummer functions) and we refer the reader to \cite{NIST} for details. In the remaining cases where analytic solutions do not appear available (prolate and oblate spheroidal coordinates in sections~\ref{sec:elliptic1} and \ref{sec:elliptic2}), a numerical approach is used to approximate the solutions. In Section~\ref{sec:vortex_breakdown} we combine several solutions from three of the coordinate systems and discuss their similarity to flows with a vortex breakdown~\cite{keller95}.

\section{Beltrami Flows}
\label{sec:Beltrami}
Inviscid fluid flows are usually modeled by the Euler equation (plus relevant energy and state equations)
\begin{equation}
  \frac{\partial \bfu}{\partial t}
  +
  (\bfu\cdot\nabla)\bfu
  =
  -\frac{\nabla p}{\rho}
  -
  g\bfz,
  \label{eq:euler_equation}
\end{equation}
where $\bfu(\bfx,t)$ is the velocity of the fluid at some point $\bfx\in\Real^3$ and time $t\in\Real$, $p(\bfx,t)$ is the pressure, $\rho(\bfx,t)$ is the mass density, $g$ is the acceleration due to gravity, and $\bfz$ is the upward pointing unit coordinate vector. Gradients are taken with respect to the spatial variable $\bfx\in\Real^3$. The mass conservation equation is
\begin{equation}
  \frac{\partial\rho}{\partial t}
  +
  \nabla\cdot(\rho\bfu)
  =
 \frac{\partial\rho}{\partial t}
  +
  \bfu\cdot\nabla\rho
  +
  \rho(\nabla\cdot\bfu)
  =
  0.
  \label{eq:mass_conservation}
\end{equation}
Taking ``curl'' of the equation for velocity \eqref{eq:euler_equation}, we obtain an equation for the vorticity, $\bsomega=\nabla\times\bfu$,
\begin{equation}
  \frac{\partial\bsomega}{\partial t}
  =
  \nabla\times(\bfu\times\bsomega)
  +
  \frac{1}{\rho^2}(\nabla\rho\times\nabla p),
  \label{eq:vorticity_equation}
\end{equation}
also known as the ``vorticity'' equation. A Beltrami flow is one in which the vorticity and velocity satisfy the Beltrami condition \cite{majdabertozzi01}
\begin{equation}
  \bsomega
  =
  \alpha\bfu
  \quad\text{ for some }\alpha\ne0,
  \label{eq:beltrami}
\end{equation}
where $\alpha$ is called the abnormality and is, in general, a function of position and time. This condition implies that $\bfu\times\bsomega={\bf 0}$ everywhere.

It is shown in \cite{majdabertozzi01} that a steady Beltrami pair $\bf u$, $\boldsymbol\omega$ for which $\nabla\cdot\bfu=0$ is a solution to the Euler equation \eqref{eq:euler_equation} with constant density $\rho$; the mass conservation equation \eqref{eq:mass_conservation} is then trivially satisfied as well. The vorticity equation \eqref{eq:vorticity_equation} then implies that such a flow will be what is known in the atmospheric science community as ``barotropic'', since the ``baroclinic'' term $\nabla\rho\times\nabla p$ will vanish. Also, taking divergence of both sides of~\eqref{eq:beltrami} and using $\nabla\cdot\bfu=0$ results in a necessary condition on the function $\alpha$
\begin{equation*}
  \bfu\cdot\nabla\alpha=0,
\end{equation*}
so $\alpha$ is constant on the streamlines of the flow.

A steady Beltrami pair $\bf u$, $\boldsymbol\omega$ for which $\nabla\cdot\bfu=0$ is a solution to the governing equations with nonconstant (steady) mass density as well. From \eqref{eq:euler_equation} we can solve for $\dfrac{\nabla p}{\rho}$, and \eqref{eq:mass_conservation} implies $\bfu\cdot\nabla\rho=0$. Thus any smooth enough mass density field that is constant on the streamlines will satisfy the governing equations as well.

Finally, even if $\bf u$ is not divergence free, \eqref{eq:euler_equation} still provides a solution for $\dfrac{\nabla p}{\rho}$, and if $\rho$ can be found that satisfies the mass equation \eqref{eq:mass_conservation}, then a complete solution to the governing equations is obtained. The mass density would need to satisfy the ODE
\begin{equation*}
  \frac{d\rho}{ds}
  =
  -\rho(\nabla\cdot\bfu)
\end{equation*}
along each streamline, parameterized with the parameter $s$.

The special case of \eqref{eq:beltrami}, in which $\alpha$ is independent of the spatial variables, is known as a {\it Trkalian} flow \cite{lakhtakia94,trkal94}. In this case, $\alpha$ is also independent of time \cite{trkal94}. Our work focuses exclusively on Trkalian flows.

\section{The Bragg--Hawthorne Equation}
\label{sec:Bragg}
We now formulate the basic equations governing incompressible, steady, axisymmetric Eulerian flows with constant mass density and in the absence of body forces. Additional details can be found, for example, in \cite{batchelor,moffatt69}.

The assumptions of incompressibility, axisymmetry, and steady state allow us to use the stream function formulation in cylindrical coordinates $(r,\theta,z)$ with a stream function $\psi(r,z)$. The corresponding velocity of the flow is then
\begin{equation}
\label{eq:velocity}
  \bfu
  =
  -
  \frac{1}{r}\frac{\partial\psi}{\partial z}\,\bfr
  +
  w\,\bstheta
  +
  \frac{1}{r}\frac{\partial\psi}{\partial r}\,\bfz,
\end{equation}
where $\bfr$, $\bstheta$, and $\bfz$ are the cylindrical coordinates basis vectors, and $w(r,z)$ is the tangential component of the velocity. In this stream function formulation the incompressibility condition, $\nabla\cdot\bfu=0$, is automatically satisfied provided $\psi$ is smooth enough. The vorticity, $\bsomega=\nabla\times\bfu$, then satisfies
\begin{equation}
  \label{eq:vorticity}
  \bsomega
  =
  -
  \frac{1}{r}\frac{\partial(rw)}{\partial z}\,\bfr
  -
  \frac{1}{r}\,D^2\psi\,\boldsymbol{\theta}
  +
  \frac{1}{r}\frac{\partial(rw)}{\partial r}\,\bfz,
\end{equation}
where
\begin{equation*}
  D^2\psi
  =
  r\,\frac{\partial}{\partial r}\left(\frac{1}{r}\frac{\partial\psi}{\partial r}\right)
  +
  \frac{\partial^2\psi}{\partial z^2}
  =
  \frac{\partial^2\psi}{\partial r^2}
  -
  \frac{1}{r}
  \frac{\partial\psi}{\partial r}
  +
  \frac{\partial^2\psi}{\partial z^2}.
\end{equation*}
We let
\begin{equation}
  \label{eq:swirl_head}
  C=rw\qquad\text{and}\qquad H=\frac{p}{\rho}+\frac{1}{2}|\bfu|^2
\end{equation}be the {\it swirl} (sometimes referred to as angular momentum) and {\it energy head}, respectively, where $p(r,z)$ is the fluid pressure and $\rho$ its mass density. The conditions for a steady flow are \cite{batchelor}
\begin{gather*}
  C=C(\psi),
  \qquad
  H=H(\psi),
\end{gather*}
i.e., both the swirl and the head are constant on the stream surfaces $\psi=\text{const.}$, and hence along streamlines. The momentum equation~\eqref{eq:euler_equation} can be written as
\begin{equation}
  \label{eq:bragg-hawthorne}
  D^2\psi
  =
  r^2\,\frac{dH}{d\psi}
  -
  C\,\frac{dC}{d\psi},
\end{equation}
which is known as the Bragg--Hawthorne equation.

When the energy head is constant in space so that $dH/d\psi=0$, equation \eqref{eq:bragg-hawthorne} reduces to $D^2\psi=-C\,\dfrac{dC}{d\psi}$ and $-\dfrac{1}{r}\,D^2\psi=w\,\dfrac{dC}{d\psi}$. The vorticity equation \eqref{eq:vorticity} can then be rewritten as
\begin{equation*}
  \bsomega
  =
  \frac{dC}{d\psi}
  \left(
    -
    \frac{1}{r}\frac{\partial\psi}{\partial z}\,\bfr
    +
    w\,\boldsymbol{\theta}
    +
    \frac{1}{r}\frac{\partial\psi}{\partial r}\,\bfz
  \right)
  =
  \frac{dC}{d\psi}\,\bfu,
\end{equation*}
so the corresponding flow is a {\it Beltrami} flow, in which velocity and vorticity are parallel everywhere in space.

In what follows, we will focus on the simple case
\begin{equation}
  \label{eq:C=alpha_psi}
  C(\psi)=\alpha\psi
\end{equation}
studied, for example, in \cite{moffatt69}. In this case, $\bsomega=\alpha\bfu$, so the proportionality constant is independent of $r$ and $z$ and the flow is Trkalian. The Bragg--Hawthorne equation~\eqref{eq:bragg-hawthorne} then further reduces to
\begin{equation}
\label{eq:eqn_psi}
  D^2\psi
  =
  -\alpha^2\psi.
\end{equation}
Due to the assumption of axisymmetry, the stream function $\psi(r,z)$ solving \eqref{eq:eqn_psi} is sought in the half plane $\{(r,z):\ r>0,\ z\in\Real\}$ with a constant boundary condition on the $z$-axis in order for the axis to be a streamline. Additionally, we set the constant to be zero, since otherwise $C(\psi)=\alpha\psi$ would imply nonzero swirl $C$ on the $z$-axis, and the expression $C=rw$ would lead to an unbounded azimuthal component $w$ of the velocity on the $z$-axis. The condition that $\psi=0$ for $r=0$ will be used in the next sections to specify boundary conditions on the functions arising by separating variables.

\subsection{Symmetry Solutions}
\label{sec:symmetry_solutions}
Note that equation~\eqref{eq:eqn_psi} possesses several properties that can be used to construct additional Beltrami solutions:
\begin{enumerate}
  \item By linearity and homogeneity of \eqref{eq:eqn_psi}, a linear combination of two solutions is another solution.
  \item If $\psi(r,z)$ is a solution, then so are $\psi(r,-z)$ and $\psi(r,z+\zeta)$ for any $\zeta\in\Real$. This follows since only a second derivative in $z$ is present in the operator $D^2$.
  \item If $\psi(r,z)$ is a solution and $\psi(r,-z)=\psi(r,z)$, then $\Psi(r,z)=\psi(r,z-\zeta)-\psi(r,z+\zeta)$ is a solution that satisfies $\Psi(r,0)=0$ and $\Psi(r,-z)=-\Psi(r,z)$ for any $\zeta\in\Real$.
  \item If $\psi(r,z)$ is a solution and $\psi(r,-z)=-\psi(r,z)$, then $\Psi(r,z)=\psi(r,z-\zeta)+\psi(r,z+\zeta)$ is a solution that satisfies $\Psi(r,0)=0$ and $\Psi(r,-z)=-\Psi(r,z)$ for any $\zeta\in\Real$.
  \item If $\psi(r,z)$ is a solution, then $\Psi(r,z)=\psi(r,\zeta+z)-\psi(r,\zeta-z)$ is a solution that satisfies $\Psi(r,0)=0$ and $\Psi(r,-z)=-\Psi(r,z)$ for any $\zeta\in\Real$.
\end{enumerate}
The last three properties can be used to construct flows in the upper half space that satisfy the boundary condition that the flow cannot penetrate the ground.

\section{Equations in Various Coordinate Systems}
\label{sec:coordinates}
In this section we explore the various appropriate coordinate systems, rewrite the governing equation \eqref{eq:eqn_psi} in these systems, and, under the assumption of separability of the sought solution, rewrite the problem as a system of two independent ODEs to be solved.

Motivated by the schematics of a vortex breakdown in a tornadic flow shown in Fig.~\ref{fig:vortex_breakdown}, we will focus on coordinate systems that seem most ``friendly'' to modeling such flows. Specifically, after discussing cylindrical and spherical coordinate systems, we will focus on paraboloidal and prolate and oblate spheroidal coordinate systems.
\begin{figure}
  \includegraphics[width=0.9\textwidth]{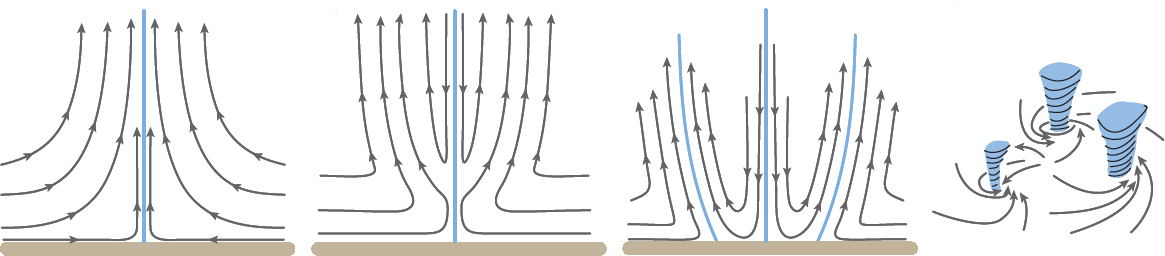}
  \caption{Schematic examples of tornadic flows illustrating a vortex breakdown. A single-cell updraft flow (left panel) bifurcates into two-cell flows (middle two panels) and further bifurcates into several vortices (right panel).}
  \label{fig:vortex_breakdown}
\end{figure}

\subsection{Cylindrical Coordinates}
\label{sec:cylindrical}
Assume that the solution to \eqref{eq:eqn_psi} has the form $\psi(r,z)=f(r)g(z)$ with $r\ge0$ and $-\infty<z<\infty$. With primes denoting the derivatives with respect to the appropriate variables, equation \eqref{eq:eqn_psi} can be written as
\begin{equation*}
  f''g
  -
  \frac{1}{r}\,f'g
  +
  fg''
  =
  -\alpha^2fg,
\end{equation*}
which can be separated into the equations
\begin{gather*}
  f''
  -
  \frac{1}{r}\,f'
  +
  cf
  =
  0,\\
  g''
  +
  (\alpha^2-c)\,g
  =
  0
\end{gather*}
for some $c\in\Real$.

Recall that the boundary condition on the stream function is $\psi(0,z)=f(0)g(z)=0$ which implies that $f(0)=0$. With this boundary condition $\psi$ is a $\mathcal{C}^\infty$ function on $\{(r,z):\ r>0\}$ and velocity (and vorticity) components can be obtained using \eqref{eq:velocity}, \eqref{eq:swirl_head} and \eqref{eq:C=alpha_psi}.

Note that if we define a function $F(\mu)$ with $\mu=r^2$ via $f(r)=F(\mu)$, we can rewrite the system as
\begin{gather*}
  F''
  +
  \frac{c}{4\mu}\,F
  =
  0,\\
  g''
  +
  (\alpha^2-c)g
  =
  0,
\end{gather*}
to be solved for $\mu>0$ and $-\infty<z<\infty$, where the boundary condition becomes $F(0)=0$. As we see below, this system has solutions for all values of the constant $c\in\Real$.

\subsubsection{Case $\alpha^2-c=0$}
In this case we immediately have $g(z)=a+bz$ as a solution for $z\in\Real$. The equation for $f$ now has the general solution $f(r)=c_1rJ_1(|\alpha|r)+c_2rY_1(|\alpha|r)$, where $J_1$ an $Y_1$ are the Bessel functions of the first and second kind, respectively. The initial condition $f(0)=0$ forces $c_2=0$, and we have the solution for the stream function
\begin{equation*}
  \psi(r,z)
  =
  (a+bz)rJ_1(|\alpha|r).
\end{equation*}

\subsubsection{Case $\alpha^2-c<0$}
Note that in this case $c>0$ and we have the solution $f(r)=c_1rJ_1(r\sqrt{c})$ and $g(z)=ae^{z\sqrt{c-\alpha^2}}+be^{-z\sqrt{c-\alpha^2}}$.
The stream function is then given by
\begin{equation*}
  \psi(r,z)
  =
  rJ_1(r\sqrt{c})\left(ae^{z\sqrt{c-\alpha^2}}+be^{-z\sqrt{c-\alpha^2}}\right).
\end{equation*}

\subsubsection{Case $\alpha^2-c>0$}
The general solution for $g$ is $g(z)=a\sin{\left(z\sqrt{\alpha^2-c}\right)}+b\cos{\left(z\sqrt{\alpha^2-c}\right)}$ and the solution for $f$ depends on the sign of $c$. If $c=0$, then $f(r)=c_1r^2$. If $c>0$, then, as above, $f(r)=c_1rJ_1(r\sqrt{c})$. Finally, if $c<0$, then $f(r)=c_1irJ_1(ir\sqrt{-c})$, which is a real-valued function. The stream function can then take one of the three forms,
\begin{equation*}
  \psi(r,z)
  =
  \left(a\sin{\left(z\sqrt{\alpha^2-c}\right)}+b\cos{\left(z\sqrt{\alpha^2-c}\right)}\right)
  \cdot
  \begin{cases}
    rJ_1(r\sqrt{c})    & \text{ for } c > 0, \\
    r^2                & \text{ for } c = 0, \\
    irJ_1(ir\sqrt{-c}) & \text{ for } c < 0.
  \end{cases}
\end{equation*}

It follows from the above results that there are three types of structures that these solutions have as shown in Fig.~\ref{fig:psi_cylindrical}.
\begin{figure}
  \includegraphics[width=0.32\textwidth]{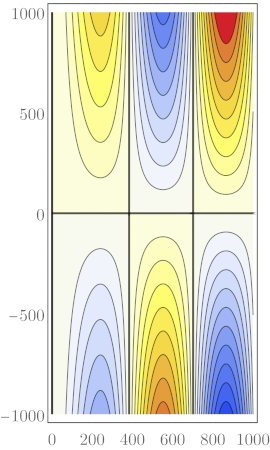}
  \hfill
  \includegraphics[width=0.32\textwidth]{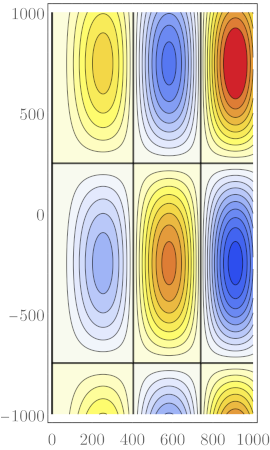}
  \hfill
  \includegraphics[width=0.32\textwidth]{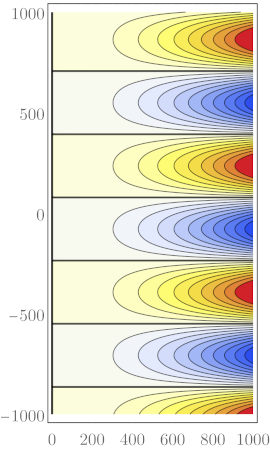}
  \vskip\baselineskip
  \caption{Contour plots of the stream functions $\psi$ obtained using cylindrical coordinates. Case 4.1.1, qualitatively similar to Case 4.1.2 (left); Case 4.1.3 with $0<c<\alpha^2$ (center); and Case 4.1.3 with $c<0$, qualitatively similar to Case 4.1.3 with $c=0$ (right).}
  \label{fig:psi_cylindrical}
\end{figure}
When $c\ge\alpha^2$, the contour plots of $\psi$ consist of vertical infinite or semi-infinite strips separated by vertical lines where $\psi=0$ (left panel). When $0<c<\alpha^2$, the contour plots consist of rectangular blocks (middle panel). When $c\le0$, the contour plots consist of horizontal semi-infinite strips (right panel).

In all displayed contour plots, the horizontal axis is the $r$-axis and the vertical axis is the $z$-axis. The thicker contours indicate where $\psi=0$. The color coding distinguishes between regions in the $rz$-plane where the flow is clockwise or counterclockwise, and it also corresponds to the tangential component of the flow being either into the $rz$-plane or out of it. In all plots, we use $\alpha=0.01$. This choice is affected by the observation that in some tornadic flows with velocity on the order of tens of meters per second the vorticity is on the order of tenths per second. The displayed window is motivated by a tornado scale and the length units can be thought of as meters.

\subsection{Spherical Coordinates}
\label{sec:spherical}
We next consider the spherical coordinates
\begin{gather*}
  r
  =
  R\sin\phi,\\
  z
  =
  R\cos\phi,\\
  R\ge0,\ 0\le\phi\le\pi.
\end{gather*}
In this coordinate system, equation \eqref{eq:eqn_psi} becomes
\begin{equation*}
  \frac{\partial^2\psi}{\partial R^2}
  +
  \frac{1}{R^2}\frac{\partial^2\psi}{\partial\phi^2}
  -
  \frac{\cot\phi}{R^2}\frac{\partial\psi}{\partial\phi}
  =
  -\alpha^2\psi.
\end{equation*}
Under the assumption of separability, $\psi(R,\phi)=f(R)g(\phi)$, this equation becomes
\begin{equation*}
  f''g
  +
  \frac{1}{R^2}\,fg''
  -
  \frac{\cot\phi}{R^2}\,fg'
  =
  -\alpha^2fg,
\end{equation*}
which can be separated into the equations
\begin{gather*}
  R^2f''
  +
  (\alpha^2R^2-c)\,f
  =
  0,\\
  g''
  -
  (\cot\phi)\,g'
  +
  cg
  =
  0
\end{gather*}
for some $c\in\Real$.

The boundary condition on the stream function, $\psi(0,z)=0$, immediately implies that $g(0)=g(\pi)=0$. Additionally, if $f(0)\ne0$ and $g\not\equiv0$, then $\psi$ is not continuous at the origin as can be seen by computing limits of $f(R)g(\phi)$ as $R\to0$ with various values of $\phi$. Therefore, the required boundary conditions are
\begin{equation*}
  f(0)=0\qquad\text{and}\qquad g(0)=g(\pi)=0.
\end{equation*}
With these boundary conditions $\psi$ is a $\mathcal{C}^\infty$ function on $\{(r,z):\ r>0\}$. The problem for $f$ is underdetermined with a regular singular point $R=0$, and for $g$ we have a second-order boundary value eigenvalue problem with regular singular endpoints.

Note that if we define a function $G(\eta)$ with $\eta=\cos\phi$ for $0\le\phi\le\pi$ via $g(\phi)=G(\eta)$, we can rewrite the system as
\begin{equation}
  f''
  +
  \left(\alpha^2-\frac{c}{R^2}\right)\,f
  =
  0,
  \qquad
  G''
  +
  \frac{c}{1-\eta^2}\,G
  =
  0,
  \label{eq:bvp_spherical}
\end{equation}
to be solved for $R>0$ and $-1<\eta<1$, where the boundary conditions become
\begin{equation*}
  f(0)=0\qquad\text{and}\qquad G(-1)=G(1)=0.
\end{equation*}

We now have the following lemma describing the spectrum of the eigenvalue problem for $G$ (or $g$).
\begin{lemma}
  The system \eqref{eq:bvp_spherical} with $(R,\eta)\in(0,\infty)\times(-1,1)$ and the boundary conditions $f(0)=G(\pm1)=0$ has nontrivial solutions if and only if $c=m(m+1)$ for $m\in\N$.
  \label{lem:spherical}
\end{lemma}
\begin{proof}
When $c=0$, the second equation in \eqref{eq:bvp_spherical} has only the trivial solution, so let's assume $c\ne0$. The general solution can then be written in terms of the hypergeometric function ${}_2{\mathcal F}_1$
\begin{equation*}
  G(\eta)
  =
  c1\cdot{}_2{\mathcal F}_1\left(-\frac{d_1}{4},\frac{d_2}{4},\frac12,\eta^2\right)
  +
  c_2\eta\cdot{}_2{\mathcal F}_1\left(-\frac{d_2}{4},\frac{d_1}{4},\frac32,\eta^2\right),
\end{equation*}
where $c_1,c_2\in\Real$ and $d_1=\sqrt{1+4c}+1$ and $d_2=\sqrt{1+4c}-1$. We now have
\begin{equation}
  \begin{split}
    \lim_{\eta\to1}{}_2{\mathcal F}_1\left(-\frac{d_1}{4},\frac{d_2}{4},\frac12,\eta^2\right)
    &=
    \frac{\sqrt{\pi}}{\Gamma\left(\dfrac{3-\sqrt{1+4c}}{4}\right)\Gamma\left(\dfrac{3+\sqrt{1+4c}}{4}\right)},\\
    \lim_{\eta\to1}{}_2{\mathcal F}_1\left(-\frac{d_2}{4},\frac{d_1}{4},\frac32,\eta^2\right)
    &=
    -\frac{2\sqrt{\pi}}{c\,\Gamma\left(\dfrac{1-\sqrt{1+4c}}{4}\right)\Gamma\left(\dfrac{1+\sqrt{1+4c}}{4}\right)}.
  \end{split}
  \label{eq:spherical_limits_1}
\end{equation}
These limiting values will be $0$ at the poles of the $\Gamma$ function, i.e., when the arguments of $\Gamma$ are non-positive integers. In all other cases, the limiting values of the first fundamental solution at $\pm1$ are some nonzero value $a$, while the limits of the second fundamental solution will be $\pm b$ for some nonzero $b$. No linear combination of these solution will satisfy both boundary conditions. When $c<-1/4$, all of the arguments in the $\Gamma$ functions are non-real, so we can assume $c\ge-1/4$. Any such $c$ can be written as a product $c=m(m+1)$ for some $m\in\Real$. In this case, the limits in \eqref{eq:spherical_limits_1} can be rewriten as
\begin{gather*}
  \lim_{\eta\to1}{}_2{\mathcal F}_1\left(-\frac{m+1}{2},\frac{m}{2},\frac12,\eta^2\right)
  =
  \frac{\sqrt{\pi}}{\Gamma\left(\dfrac{1-m}{2}\right)\Gamma\left(\dfrac{m+2}{2}\right)},\\
  \lim_{\eta\to1}{}_2{\mathcal F}_1\left(-\frac{m}{2},\frac{m+1}{2},\frac32,\eta^2\right)
  =
  -\frac{2\sqrt{\pi}}{m(m+1)\,\Gamma\left(-\dfrac{m}{2}\right)\Gamma\left(\dfrac{m+1}{2}\right)},
\end{gather*}
and the requirement of at least one of the two limits being $0$ implies $m\in\Z$. Since negative values of $m$ produce the same $c=m(m+1)$ as nonnegative ones, and since $m=0$ corresponds to $c=0$, we have that only $m\in\N$ need to be considered to produce nontrivial solutions for $G$. These solutions are, up to a multiplicative constant,
\begin{equation}
  G(\eta)
  =
  \begin{cases}
    {}_2{\mathcal F}_1\left(-\dfrac{m+1}{2},\dfrac{m}{2},\dfrac12,\eta^2\right)       & \text{ for $m$ odd,} \\
    \eta\cdot{}_2{\mathcal F}_1\left(-\dfrac{m}{2},\dfrac{m+1}{2},\dfrac32,\eta^2\right) & \text{ for $m$ even},
  \end{cases}
  \label{eq:spherical_G_eta}
\end{equation}
and it follows from the definition of ${}_2{\mathcal F}_1$ that they are all polynomials (see also \cite{marsh96}).

For any $m\in\N$ and $c=m(m+1)$, the equation for $f$ in \eqref{eq:bvp_spherical} together with its boundary condition has a nontrivial solution in terms of the Bessel function of the first kind, again up to a multiplicative constant,
\begin{equation}
  f(R)
  =
  \sqrt{R}\,J_\beta(\alpha R)
  \quad\text{ with }\quad
  \beta=\frac{2m+1}{2}.
  \label{eq:spherical_f_R}
\end{equation}
\end{proof}

Using \eqref{eq:spherical_G_eta} and \eqref{eq:spherical_f_R}, the stream function in spherical coordinates has the form
\begin{equation*}
  \psi(R,\eta)
  =
  \sqrt{R}\,J_\beta(\alpha R)
  \cdot
  \begin{cases}
    {}_2{\mathcal F}_1\left(-\dfrac{m+1}{2},\dfrac{m}{2},\dfrac12,\eta^2\right)       & \text{ for $m$ odd,} \\
    \eta\cdot{}_2{\mathcal F}_1\left(-\dfrac{m}{2},\dfrac{m+1}{2},\dfrac32,\eta^2\right) & \text{ for $m$ even},
  \end{cases}
\end{equation*}
where $m\in\N$ and $\beta=(2m+1)/2$. In the original cylindrical coordinates $(r,z)$, we have
\begin{equation*}
  \psi(r,z)
  =
  (r^2+z^2)^{1/4}\,J_\beta(\alpha\sqrt{r^2+z^2})
  \cdot
  \begin{cases}
    {}_2{\mathcal F}_1\left(-\dfrac{m+1}{2},\dfrac{m}{2},\dfrac12,\dfrac{z^2}{r^2+z^2}\right)       & \text{ for $m$ odd,} \\
    \dfrac{z}{\sqrt{r^2+z^2}}\cdot{}_2{\mathcal F}_1\left(-\dfrac{m}{2},\dfrac{m+1}{2},\dfrac32,\dfrac{z^2}{r^2+z^2}\right) & \text{ for $m$ even}.
  \end{cases}
\end{equation*}

Contour plots of some of these stream functions have been shown in literature (see, e.g., \cite{marsh96}), but we include them in Fig.~\ref{fig:psi_spherical} for completeness and comparison to solutions in other coordinate systems.
\begin{figure}
  \includegraphics[width=0.32\textwidth]{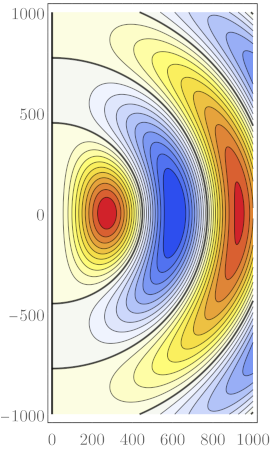}
  \hfill
  \includegraphics[width=0.32\textwidth]{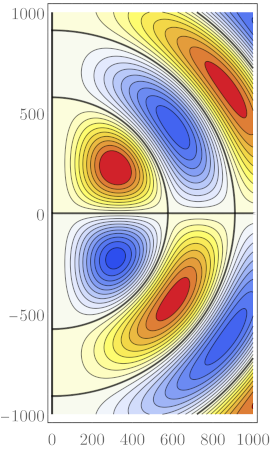}
  \hfill
  \includegraphics[width=0.32\textwidth]{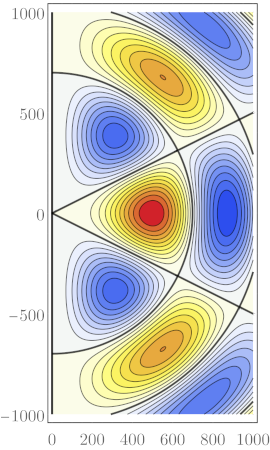}
  \vskip\baselineskip
  \includegraphics[width=0.32\textwidth]{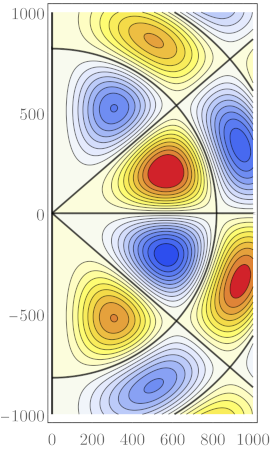}
  \hfill
  \includegraphics[width=0.32\textwidth]{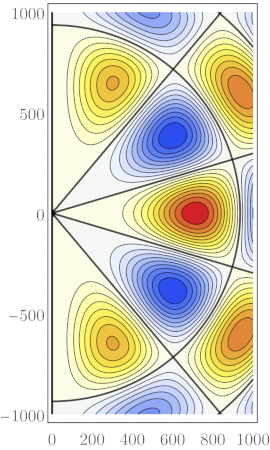}
  \hfill
  \includegraphics[width=0.32\textwidth]{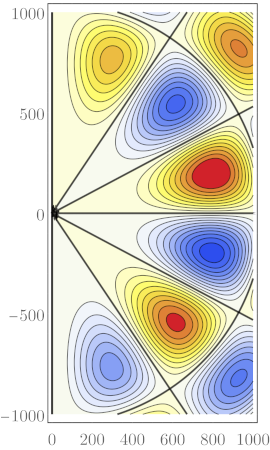}
  \caption{Contour plots of the stream functions $\psi$ obtained using spherical coordinates. Results for the first six eigenmodes ($m=1,\dots,6$) are shown left to right and top to bottom.}
  \label{fig:psi_spherical}
\end{figure}

\subsection{Paraboloidal Coordinates}
\label{sec:paraboloidal}
The paraboloidal coordinates are
\begin{gather*}
  r
  =
  uv,\\
  z
  =
  \frac{1}{2}(v^2-u^2),\\
  u\ge0,\ v\ge0.
\end{gather*}
The curves along which $u$ and $v$ are constant are shown in Fig.~\ref{fig:ell-par-cyl} on the left.
\begin{figure}
  \begin{center}
    \includegraphics[width=0.3\textwidth]{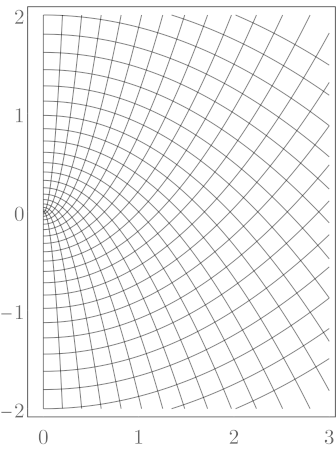}
    \hfill
    \includegraphics[width=0.3\textwidth]{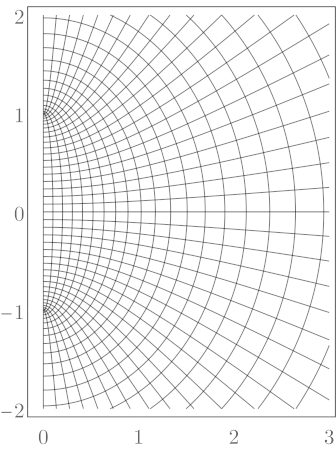}
    \hfill
    \includegraphics[width=0.3\textwidth]{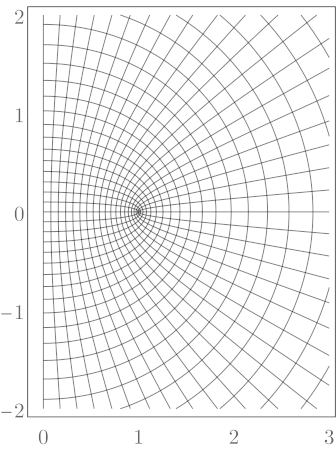}
  \end{center}
  \caption{Visualization of the paraboloidal (left), prolate spheroidal (center) and oblate spheroidal (right) coordinate systems in the $rz$-plane, the last two both with $a=1$. In the case of the paraboloidal coordinates, parabolas opening up correspond to $u$ constant and parabolas opening down to $v$ constant. For both spheroidal systems, ellipses correspond to $u$ constant and hyperbolas to $v$ constant.}
  \label{fig:ell-par-cyl}
\end{figure}
The curves with constant $u$ are the parabolas opening up and the curves with constant $v$ are the parabolas opening down. It is easy to see that
\begin{equation*}
  u^2
  =
  \sqrt{r^2+z^2}-z,
  \qquad
  v^2
  =
  \sqrt{r^2+z^2}+z.
\end{equation*}

In this coordinate system, equation \eqref{eq:eqn_psi} becomes
\begin{equation*}
  \frac{1}{u^2+v^2}
  \left(
    \frac{\partial^2\psi}{\partial u^2}
    +
    \frac{\partial^2\psi}{\partial v^2}
    -
    \frac{1}{u}\,\frac{\partial\psi}{\partial u}
    -
    \frac{1}{v}\,\frac{\partial\psi}{\partial v}
  \right)
  =
  -\alpha^2\psi.
\end{equation*}
Under the assumption of separability, $\psi(u,v)=f(u)g(v)$, becomes
\begin{equation*}
  f''g
  +
  fg''
  -
  \frac{1}{u}\,f'g
  -
  \frac{1}{v}\,fg'
  =
  -\alpha^2(u^2+v^2)\,fg,
\end{equation*}
which can be separated into the equations
\begin{gather*}
  f''
  -
  \frac{1}{u}\,f'
  +
  (\alpha^2u^2-c)\,f
  =
  0,\\
  g''
  -
  \frac{1}{v}\,g'
  +
  (\alpha^2v^2+c)\,g
  =
  0
\end{gather*}
for $c\in\Real$.

The boundary condition on the stream function, $\psi(0,z)=0$, immediately implies
\begin{equation*}
  f(0)=0\qquad\text{and}\qquad g(0)=0.
\end{equation*}
With these boundary conditions $\psi$ is a $\mathcal{C}^\infty$ function on $\{(r,z):\ r>0\}$. Hence both problems for $f$ and $g$ are underdetermined, and both problems have a regular singular point at $0$.

Note that if we define functions $F(\mu)$ and $G(\eta)$ with $\mu=u^2$ and $\eta=v^2$ via $f(u)=F(\mu)$ and $g(v)=G(\eta)$, we can rewrite the system as
\begin{gather}
  F''
  +
  \frac{\alpha^2\mu-c}{4\mu}\,F
  =
  0,
  \qquad
  G''
  +
  \frac{\alpha^2\eta+c}{4\eta}\,G
  =
  0,
  \label{eq:FG_parabolic}
\end{gather}
to be solved for $\mu>0$ and $\eta>0$, where the boundary conditions become
\begin{equation*}
  F(0)=0\qquad\text{and}\qquad G(0)=0.
\end{equation*}
The special case when $c=0$ immediately results in
\begin{equation}
  F(\mu)
  =
  c_1\sin\left(\frac{\alpha\mu}{2}\right),
  \qquad
  G(\eta)
  =
  c_2\sin\left(\frac{\alpha\eta}{2}\right).
  \label{eq:paraboloidal_special}
\end{equation}
The general solution to \eqref{eq:FG_parabolic} in the case $c\ne0$ can be written using the hypergeometric functions ${}_1{\mathcal F}_1$ and $U$
\begin{gather*}
  F(\mu)
  =
  \mu\,e^{-\frac12i|\alpha|\mu}
    \left[
      c_1\cdot{}_1{\mathcal F}_1\left(1+\frac{c}{4i|\alpha|},2,i|\alpha|\mu\right)
      +
      c_3\,U\left(1+\frac{c}{4i|\alpha|},2,i|\alpha|\mu\right)
    \right],\\
  G(\eta)
  =
  \eta\,e^{-\frac12i|\alpha|\eta}
    \left[
      c_2\cdot{}_1{\mathcal F}_1\left(1-\frac{c}{4i|\alpha|},2,i|\alpha|\eta\right)
      +
      c_4\,U\left(1-\frac{c}{4i|\alpha|},2,i|\alpha|\eta\right)
    \right],
\end{gather*}
where $c_i\in\Real$ for $i=1,\dots,4$. We now have
\begin{equation*}
    \lim_{\mu\to0}F(\mu)
    =
    \frac{4c_3}{c\,\Gamma\left(\dfrac{c}{4i|\alpha|}\right)},
    \qquad
    \lim_{\eta\to0}G(\eta)
    =
    -\frac{4c_4}{c\,\Gamma\left(-\dfrac{c}{4i|\alpha|}\right)},
\end{equation*}
and therefore to satisfy the boundary conditions we need to have $c_3=c_4=0$. The solution for $c\ne0$ is then
\begin{equation}
  \begin{split}
    F(\mu)
    =
    c_1\,\mu\,e^{-\frac12i|\alpha|\mu}{}_1{\mathcal F}_1\left(1+\frac{c}{4i|\alpha|},2,i|\alpha|\mu\right),\quad c_1\in\Real,\\
    G(\eta)
    =
    c_2\,\eta\,e^{-\frac12i|\alpha|\eta}{}_1{\mathcal F}_1\left(1-\frac{c}{4i|\alpha|},2,i|\alpha|\eta\right),\quad c_2\in\Real.
  \end{split}
  \label{eq:paraboloidal_solution}
\end{equation}
\begin{remark}
  By using (13.2.4) and (13.4.1) in \cite{NIST} and trigonometric identities, it can be shown that the solutions in \eqref{eq:paraboloidal_solution} are real valued and, in fact, are given by
  \begin{gather*}
    F(\mu)
    =
    C_1\,\mu\int_0^1\cos\left[\frac{c}{4|\alpha|}\log\left(\frac{1-u}{u}\right)+|\alpha|\mu\left(u-\frac12\right)\right]du,\\
    G(\eta)
    =
    C_2\,\eta\int_0^1\cos\left[-\frac{c}{4|\alpha|}\log\left(\frac{1-u}{u}\right)+|\alpha|\eta\left(u-\frac12\right)\right]du,
  \end{gather*}
  where $C_i=c_i\,\dfrac{4|\alpha|}{c\,\pi}\sinh\left(\dfrac{c\,\pi}{4|\alpha|}\right)$ for $i=1,2$.
\end{remark}

Using \eqref{eq:paraboloidal_solution}, the stream function in paraboloidal coordinates has the form
\begin{equation}
  \psi(u,v)
  =
  u^2v^2\,e^{-\frac12i|\alpha|(u^2+v^2)}{}_1{\mathcal F}_1\left(1+\frac{c}{4i|\alpha|},2,i|\alpha|u^2\right){}_1{\mathcal F}_1\left(1-\frac{c}{4i|\alpha|},2,i|\alpha|v^2\right).
  \label{eq:psi_paraboloidal}
\end{equation}
When $c=0$, this corresponds to
\begin{equation*}
  \psi(u,v)
  =
  \sin\left(\frac{\alpha u^2}{2}\right)\sin\left(\frac{\alpha v^2}{2}\right),
\end{equation*}
which agrees with the special case solution given in \eqref{eq:paraboloidal_special}.

Contour plots of three stream functions obtained in paraboloidal coordinates are shown in Fig.~\ref{fig:psi_paraboloidal}.
\begin{figure}
  \includegraphics[width=0.32\textwidth]{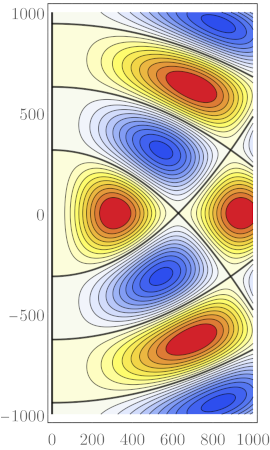}
  \hfill
  \includegraphics[width=0.32\textwidth]{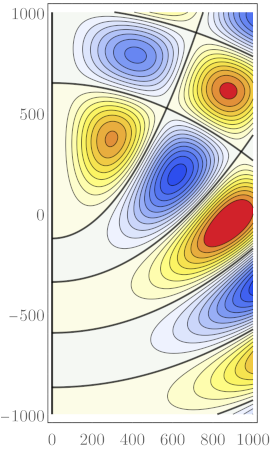}
  \hfill
  \includegraphics[width=0.32\textwidth]{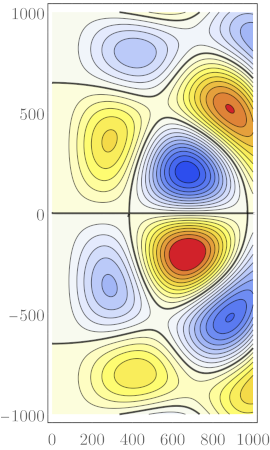}
  \caption{Contour plots of the stream functions $\psi$ obtained using paraboloidal coordinates. The left and middle panels correspond to \eqref{eq:psi_paraboloidal} with $c=0$ and $c=-0.05$, respectively. The right panel shows the stream function obtained from \eqref{eq:psi_paraboloidal} with $c=-0.05$ by applying transformation (5) from Section~\ref{sec:symmetry_solutions} with $\zeta=0$.}
  \label{fig:psi_paraboloidal}
\end{figure}
On the left is the case with $c=0$ and the middle panel corresponds to $c=-0.05$. It is clear from \eqref{eq:psi_paraboloidal} and the definition of the paraboloidal coordinates that changing the sign of $c$ corresponds to interchanging the roles of $u$ and $v$ and therefore changing the sign of $z$. Consequently, a contour plot for $c=0.05$ can be obtained from that with $c=-0.05$. Since the flows with $c\ne0$ do not appear to have any kind of symmetry with respect to the $z=0$ plane, they can be used, together with transformation (5) from Section~\ref{sec:symmetry_solutions} to generate new flows for which $z=0$ is a stream surface. This is illustrated in the right panel of  Fig.~\ref{fig:psi_paraboloidal} in which \eqref{eq:psi_paraboloidal} with $c=-0.05$ was used.

Finally, we note that for large $\mu$ the solution for $F$ in \eqref{eq:FG_parabolic} resembles that of $F''+\dfrac{\alpha^2}{4}F=0$ and therefore exhibits near periodicity in $\mu$. Consequently, the contour plots of the corresponding stream function $\psi$ consist of repeating blocks along some of the curves shown in the left panel of Fig.~\ref{fig:ell-par-cyl} as indicated by the contour plots in Fig.~\ref{fig:psi_paraboloidal}.

\subsection{Prolate Spheroidal Coordinates}
\label{sec:elliptic1}
The prolate spheroidal coordinates are
\begin{align*}
  r
  =
  a\sinh u\sin v,\\
  z
  =
  a\cosh u\cos v,\\
  u\ge0,\ 0\le v\le\pi.
\end{align*}
The curves along which $u$ and $v$ are constant are shown in Fig.~\ref{fig:ell-par-cyl} in the center. The curves with constant $u$ are the hyperbolas and the curves with constant $v$ are the ellipses.

In order to obtain $u$ and $v$ from $r$ and $z$, we can use the conversion formulas from \cite{sun17}, which, in this case, result in
\begin{equation*}
  u
  =
  \frac12\log\left(1-2q+2\sqrt{q^2-q}\right),
  \qquad
  v
  =
  \begin{cases}
    \arcsin\sqrt{p}     & \text{ if } z\ge0,\\
    \pi-\arcsin\sqrt{p} & \text{ if } z<0,
  \end{cases}
\end{equation*}
where
\begin{equation*}
  p
  =
  \frac{-B+\sqrt{B^2+4a^2r^2}}{2a^2},
  \qquad
  q
  =
  \frac{-B-\sqrt{B^2+4a^2r^2}}{2a^2},
  \qquad\text{with }
  B
  =
  r^2+z^2-a^2.
\end{equation*}

In this coordinate system, equation \eqref{eq:eqn_psi} becomes
\begin{equation*}
  \frac{1}{a^2(\cosh^2u-\cos^2v)}
  \left(
    \frac{\partial^2\psi}{\partial u^2}
    +
    \frac{\partial^2\psi}{\partial v^2}
    -
    \coth u\,\frac{\partial\psi}{\partial u}
    -
    \cot v\,\frac{\partial\psi}{\partial v}
  \right)
  =
  -\alpha^2\psi.
\end{equation*}
Under the assumption of separability, $\psi(u,v)=f(u)g(v)$, becomes
\begin{equation*}
  f''g
  +
  fg''
  -
  (\coth u)\,f'g
  -
  (\cot v)\,fg'
  =
  -\alpha^2a^2(\cosh^2u-\cos^2v)\,fg,
\end{equation*}
which, with $A=\alpha a$, can be separated into the equations
\begin{gather*}
  f''
  -
  (\coth u)\,f'
  +
  (A^2\cosh^2u-c)\,f
  =
  0,\\
  g''
  -
  (\cot v)\,g'
  -
  (A^2\cos^2v-c)\,g
  =
  0
\end{gather*}
for $c\in\Real$.

Recall that the boundary condition on the stream function is $\psi(0,z)=0$. The $z$-axis consists of three intervals: $(-\infty,-a]$ corresponding to $v=\pi$, $[-a,a]$ corresponding to $u=0$, and $[a,\infty)$ corresponding to $v=0$. Therefore, necessary boundary conditions for $f$ and $g$ are
\begin{equation*}
  f(0)=0\qquad\text{and}\qquad g(0)=g(\pi)=0.
\end{equation*}
With these boundary conditions $\psi$ is a $\mathcal{C}^\infty$ function on $\{(r,z):\ r>0\}$. Hence the problem for $f$ is underdetermined with a regular singular point at $0$, and the problem for $g$ becomes a second-order boundary value eigenvalue problem with regular singular endpoints. The problem for $g$ is a regular Sturm--Louville problem~\cite{stakgold11}, and therefore there exists a countable, real, bounded below spectrum of values for $c$.

Note that if we define functions $F(\mu)$ and $G(\eta)$ with $\mu=\cosh{u}$ and $\eta=\cos{v}$ via $f(u)=F(\mu)$ and $g(v)=G(\eta)$, we can rewrite the system as
\begin{gather}
  F''
  +
  \frac{A^2\mu^2-c}{\mu^2-1}\,F
  =
  0,
  \qquad
  G''
  +
  \frac{c-A^2\eta^2}{1-\eta^2}\,G
  =
  0,
  \label{eq:FG_case1}
\end{gather}
to be solved for $\mu>1$ and $-1<\eta<1$, where the boundary conditions become
\begin{equation*}
  F(1)=0\qquad\text{and}\qquad G(-1)=G(1)=0.
\end{equation*}
Also note that in this case the two differential equations for $F$ and $G$ are actually the same, though solved on different intervals.

The easy case to solve analytically is when $c=A^2$, since in this case the differential equations reduce to $F''+A^2F=0$ and $G''+A^2G=0$. The equation for $G$, together with its boundary conditions, then forces $A=n\pi/2$ for $n\in\Z$, and we obtain two sets of solutions,
\begin{equation*}
  F(\mu)
  =
  c_1\cos\left(\frac{(2n+1)\pi}{2}\mu\right),
  \qquad
  G(\nu)
  =
  c_2\cos\left(\frac{(2n+1)\pi}{2}\nu\right)
  \qquad\text{ for }n\in\Z,
\end{equation*}
or
\begin{equation*}
  F(\mu)
  =
  c_1\sin\left(n\pi\mu\right),
  \qquad
  G(\nu)
  =
  c_2\sin\left(n\pi\nu\right)
  \qquad\text{ for }n\in\Z.
\end{equation*}

In the general case without the assumption $c=A^2$, we have not found explicit solutions analytically and instead approximated them numerically. The idea is to first approximate the eigenvalues $c$ by approximating the solution to the boundary-value problem for $G$, and once $c$ is approximated, then approximate the solution to the initial-value problem for $F$. The problem for $F$ can be supplied with a second initial condition $F'(1)=1$, because other values simply lead to different scalings of $F$, and thus of $\psi$, which does not affect the stream surfaces of $\psi$. Contour plots for various values of the parameters are shown in Figs.~\ref{fig:case1_A2_1}--\ref{fig:case1_A2_100}. In all of them $\alpha=0.01$ and $A^2$ ranges from $1$ to $100$. We note that the white lines visible along the $r$-axis for the odd eigenmodes is due to the contour plotter in {\it Mathematica} struggling with the piecewise function that converts $r$ and $z$ to $v$. The same is true in the next section.
\begin{figure}
  \includegraphics[width=0.32\textwidth]{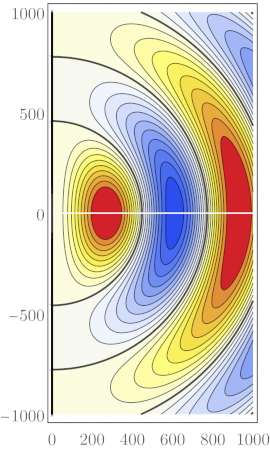}
  \hfill
  \includegraphics[width=0.32\textwidth]{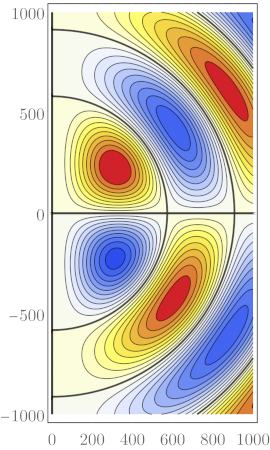}
  \hfill
  \includegraphics[width=0.32\textwidth]{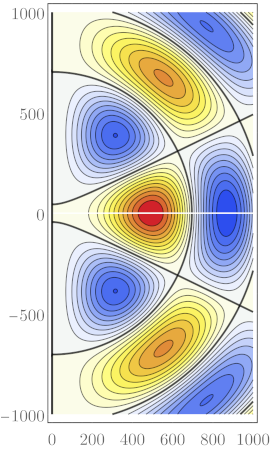}
  \vskip\baselineskip
  \includegraphics[width=0.32\textwidth]{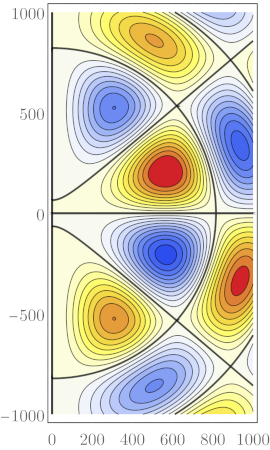}
  \hfill
  \includegraphics[width=0.32\textwidth]{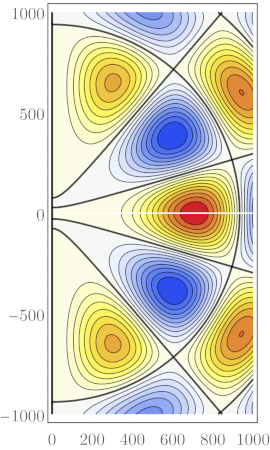}
  \hfill
  \includegraphics[width=0.32\textwidth]{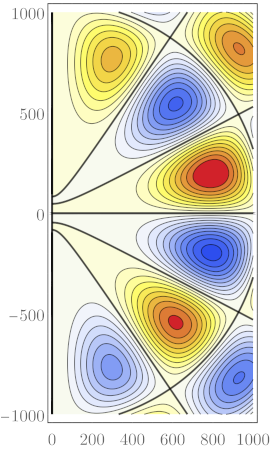}
  \caption{Contour plots of the stream functions $\psi$ obtained using prolate spheroidal coordinates. Results for the first six eigenmodes with $A^2=1$ are shown left to right and top to bottom.}
  \label{fig:case1_A2_1}
\end{figure}

\begin{figure}
  \includegraphics[width=0.32\textwidth]{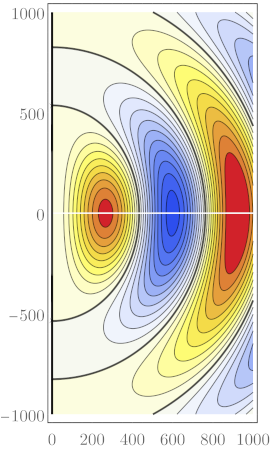}
  \hfill
  \includegraphics[width=0.32\textwidth]{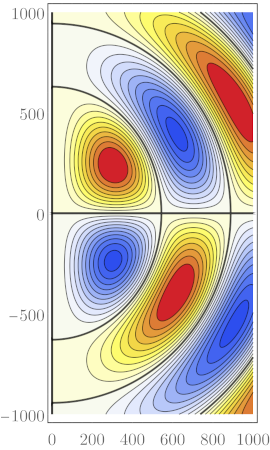}
  \hfill
  \includegraphics[width=0.32\textwidth]{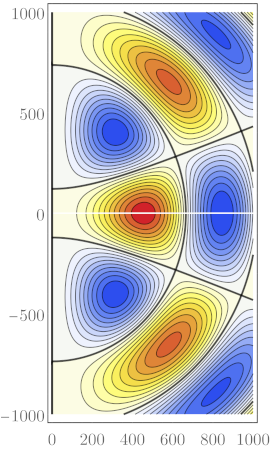}
  \vskip\baselineskip
  \includegraphics[width=0.32\textwidth]{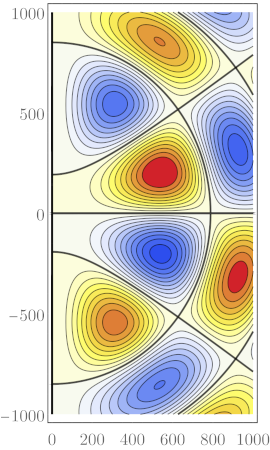}
  \hfill
  \includegraphics[width=0.32\textwidth]{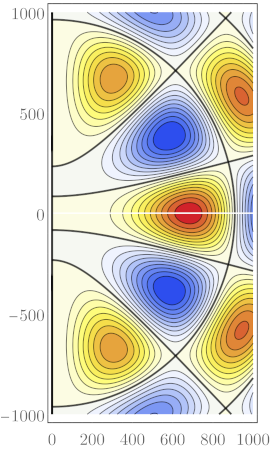}
  \hfill
  \includegraphics[width=0.32\textwidth]{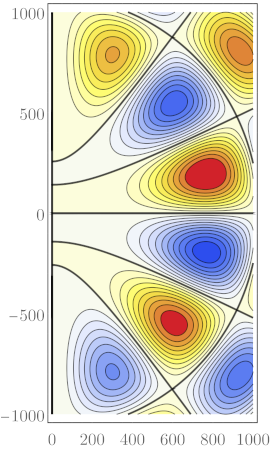}
  \caption{Contour plots of the stream functions $\psi$ obtained using prolate spheroidal coordinates. Results for the first six eigenmodes with $A^2=10$ are shown left to right and top to bottom.}
  \label{fig:case1_A2_10}
\end{figure}

\begin{figure}
  \includegraphics[width=0.32\textwidth]{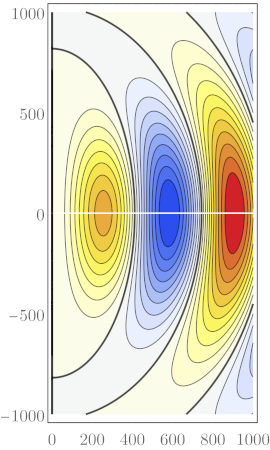}
  \hfill
  \includegraphics[width=0.32\textwidth]{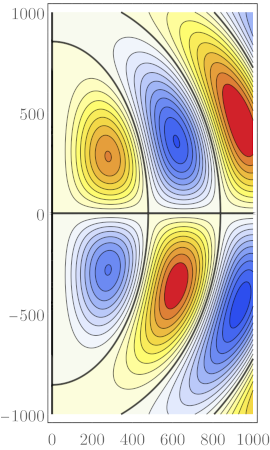}
  \hfill
  \includegraphics[width=0.32\textwidth]{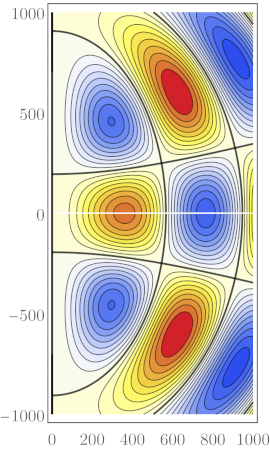}
  \vskip\baselineskip
  \includegraphics[width=0.32\textwidth]{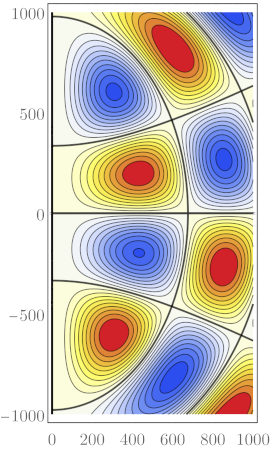}
  \hfill
  \includegraphics[width=0.32\textwidth]{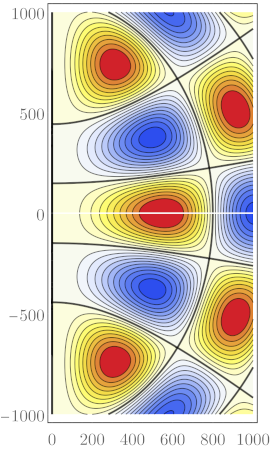}
  \hfill
  \includegraphics[width=0.32\textwidth]{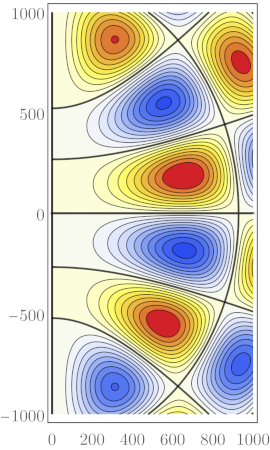}
  \caption{Contour plots of the stream functions $\psi$ obtained using prolate spheroidal coordinates. Results for the first six eigenmodes with $A^2=50$ are shown left to right and top to bottom.}
  \label{fig:case1_A2_50}
\end{figure}

\begin{figure}
  \includegraphics[width=0.32\textwidth]{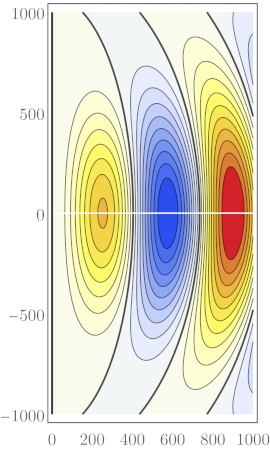}
  \hfill
  \includegraphics[width=0.32\textwidth]{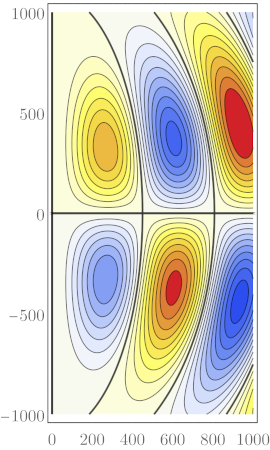}
  \hfill
  \includegraphics[width=0.32\textwidth]{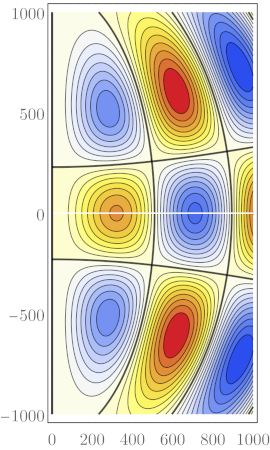}
  \vskip\baselineskip
  \includegraphics[width=0.32\textwidth]{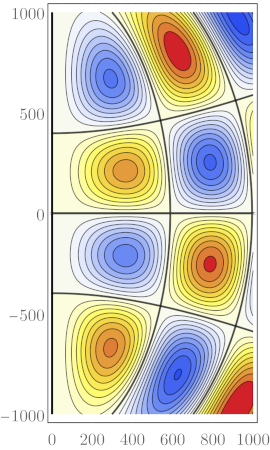}
  \hfill
  \includegraphics[width=0.32\textwidth]{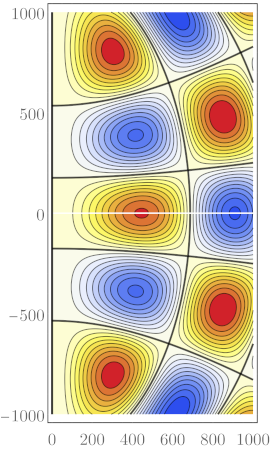}
  \hfill
  \includegraphics[width=0.32\textwidth]{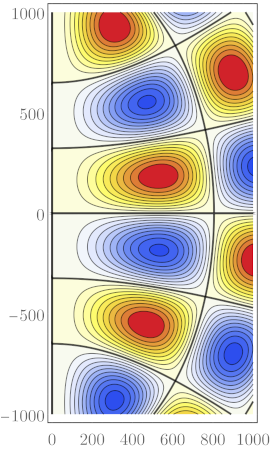}
  \caption{Contour plots of the stream functions $\psi$ obtained using prolate spheroidal coordinates. Results for the first six eigenmodes with $A^2=100$ are shown left to right and top to bottom.}
  \label{fig:case1_A2_100}
\end{figure}

We again note that for large $\mu$ the solution for $F$ in \eqref{eq:FG_case1} resembles that of $F''+A^2F=0$ and therefore exhibits near periodicity in $\mu$. Consequently, the contour plots of the corresponding stream function $\psi$ consist of repeating blocks along some of the curves shown in the middle panel of Fig.~\ref{fig:ell-par-cyl} as indicated by the contour plots in Figs.~\ref{fig:case1_A2_1}--\ref{fig:case1_A2_100}.

\subsection{Oblate Spheroidal Coordinates}
\label{sec:elliptic2}
The oblate spheroidal coordinates are
\begin{gather*}
  r
  =
  a\cosh u\cos v,\\
  z
  =
  a\sinh u\sin v,\\
  u\ge0,\ -\pi/2\le v\le\pi/2.
\end{gather*}
The curves along which $u$ and $v$ are constant are shown in Fig.~\ref{fig:ell-par-cyl} on the right. The curves with constant $u$ are the ellipses and the curves with constant $v$ are the hyperbolas.

In order to obtain $u$ and $v$ from $r$ and $z$, we can again use the conversion formulas from \cite{sun17}, which now have the form
\begin{equation*}
  u
  =
  \frac12\log\left(1-2q+2\sqrt{q^2-q}\right),
  \qquad
  v
  =
  \begin{cases}
    \ \ \ \arcsin\sqrt{p} & \text{ if } z\ge0,\\
    -\arcsin\sqrt{p}      & \text{ if } z<0,
  \end{cases}
\end{equation*}
where
\begin{equation*}
  p
  =
  \frac{-B+\sqrt{B^2+4a^2z^2}}{2a^2},
  \qquad
  q
  =
  \frac{-B-\sqrt{B^2+4a^2z^2}}{2a^2},
  \qquad\text{with }
  B
  =
  r^2+z^2-a^2.
\end{equation*}

In this coordinate system, equation \eqref{eq:eqn_psi} becomes
\begin{equation*}
  \frac{1}{a^2(\sinh^2u+\sin^2v)}
  \left(
    \frac{\partial^2\psi}{\partial u^2}
    +
    \frac{\partial^2\psi}{\partial v^2}
    -
    \tanh u\,\frac{\partial\psi}{\partial u}
    +
    \tan v\,\frac{\partial\psi}{\partial v}
  \right)
  =
  -\alpha^2\psi.
\end{equation*}
Under the assumption of separability, $\psi(u,v)=f(u)g(v)$, this equation becomes
\begin{equation*}
  f''g
  +
  fg''
  -
  (\tanh u)\,f'g
  +
  (\tan v)\,fg'
  =
  -\alpha^2a^2(\sinh^2u+\sin^2v)\,fg,
\end{equation*}
which, with $A=\alpha a$, can be separated into the equations
\begin{gather*}
  f''
  -
  (\tanh u)\,f'
  +
  (A^2\sinh^2u-c)\,f
  =
  0,\\
  g''
  +
  (\tan v)\,g'
  +
  (A^2\sin^2v+c)\,g
  =
  0,
\end{gather*}
for $c\in\Real$.

The boundary condition on the stream function is $\psi(0,z)=0$. The positive $z$-axis corresponds to $v=\pi/2$ and the negative $z$-axis corresponds to $v=-\pi/2$. Therefore, necessary boundary conditions are $g(-\pi/2)=g(\pi/2)=0$.

However, several more conditions have to be checked to ensure that $\psi$ is continuous and differentiable in $\{(r,z):\ r>0\}$. We first observe that the symmetry in the boundary-value eigenvalue problem for $g$ implies that $g$ is either an even or an odd function of $v$. The $r$-axis is divided into intervals $(0,a]$ corresponding to $u=0$ and $[a,\infty)$ corresponding to $v=0$. To ensure continuity of $\psi$ in $\{(r,z):\ r>0\}$, only continuity across the segment $(0,a]$ needs to be addressed. If $g$ is odd, there is a discontinuity in $\psi$ unless $f(0)=0$. If $g$ is even, $\psi$ is continuous across the segment without a restriction on $f(0)$.

It can be verified that $\psi$ is differentiable in the $z$-direction across the segment $(0,a]$ on the $r$-axis, so no new requirements arise there. However, in the case of an even $g$, the stream function $\psi$ is differentiable in the $r$-direction at the point $(a,0)$ on the $r$-axis only when $f'(0)=0$.

Therefore, the boundary conditions are $g(-\pi/2)=g(\pi/2)=0$ and either $f(0)=0$ when $g$ is an odd function, or $f'(0)=0$ when $g$ is an even function. Hence the problem for $f$ is underdetermined but with no singular points, and the problem for $g$ again becomes a second-order boundary-value eigenvalue problem with regular singular endpoints, whose spectrum has the same properties as in the prolate spheriodal coordinate system case.

Note that if we define functions $F(\mu)$ and $G(\eta)$ with $\mu=\sinh{u}$ and $\eta=\sin{v}$ via $f(u)=F(\mu)$ and $g(v)=G(\eta)$, we can rewrite the system as
\begin{gather}
  F''
  +
  \frac{A^2\mu^2-c}{1+\mu^2}\,F
  =
  0,
  \qquad
  G''
  +
  \frac{A^2\eta^2+c}{1-\eta^2}\,G
  =
  0,
  \label{eq:FG_case2}
\end{gather}
to be solved for $\mu>0$ and $-1<\eta<1$, where the boundary conditions become $G(-1)=G(1)=0$, and $F(0)=0$ when $G$ is odd, and $F'(0)=0$ when $G$ is even since the symmetry of $G$ is inherited from $g$.

In this coordinate system, again we have not found any explicit solutions analytically and instead approximated them numerically in the way described in the previous section. The ``missing'' initial condition for $F$ has been supplied by $F'(0)=1$ when $F(0)=0$ and by $F(0)=1$ when $F'(0)=0$. Contour plots for various values of the parameters are shown in Figs.~\ref{fig:case2_A2_1}--\ref{fig:case2_A2_100}. In all of them $\alpha=0.01$ and $A^2$ ranges from $1$ to $100$.
\begin{figure}
  \includegraphics[width=0.32\textwidth]{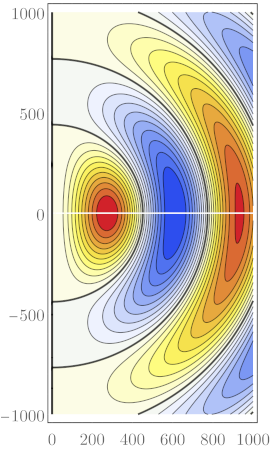}
  \hfill
  \includegraphics[width=0.32\textwidth]{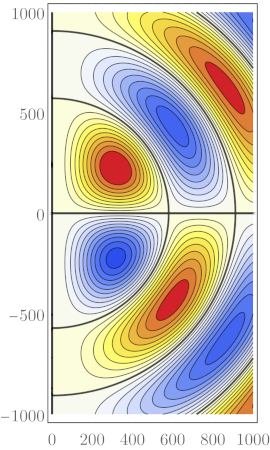}
  \hfill
  \includegraphics[width=0.32\textwidth]{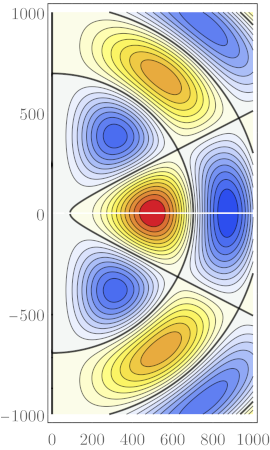}
  \vskip\baselineskip
  \includegraphics[width=0.32\textwidth]{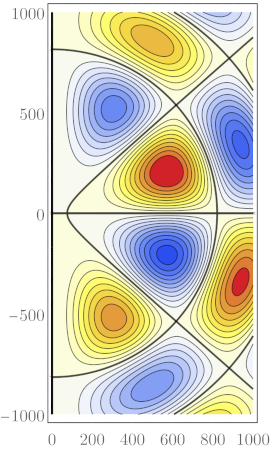}
  \hfill
  \includegraphics[width=0.32\textwidth]{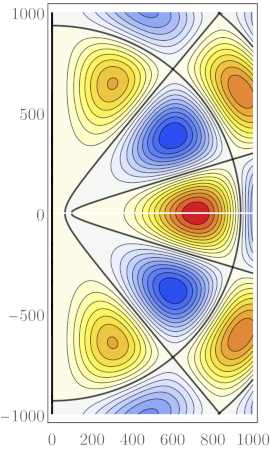}
  \hfill
  \includegraphics[width=0.32\textwidth]{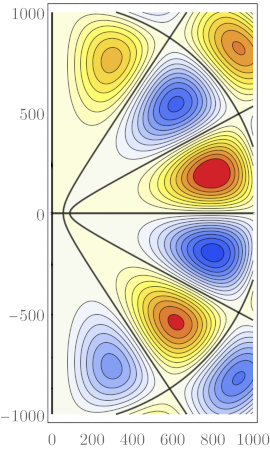}
  \caption{Contour plots of the stream functions $\psi$ obtained using oblate spheroidal coordinates. Results for the first six eigenmodes with $A^2=1$ are shown left to right and top to bottom.}
  \label{fig:case2_A2_1}
\end{figure}

\begin{figure}
  \includegraphics[width=0.32\textwidth]{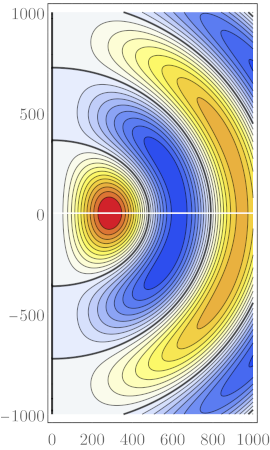}
  \hfill
  \includegraphics[width=0.32\textwidth]{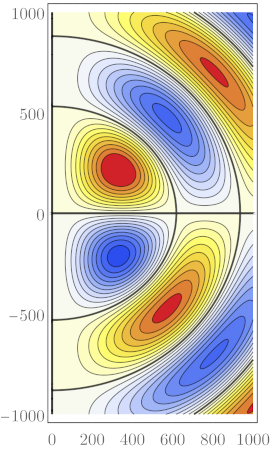}
  \hfill
  \includegraphics[width=0.32\textwidth]{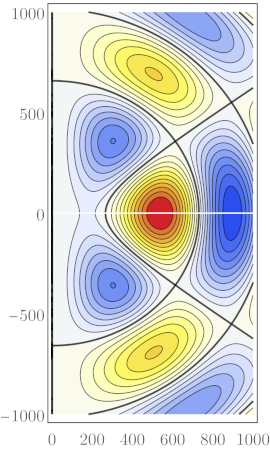}
  \vskip\baselineskip
  \includegraphics[width=0.32\textwidth]{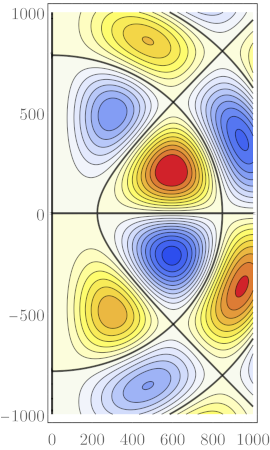}
  \hfill
  \includegraphics[width=0.32\textwidth]{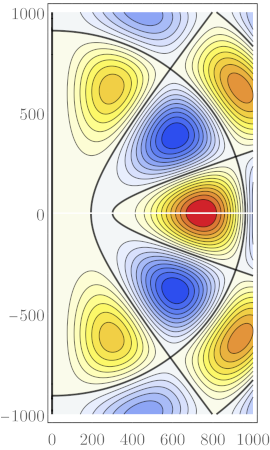}
  \hfill
  \includegraphics[width=0.32\textwidth]{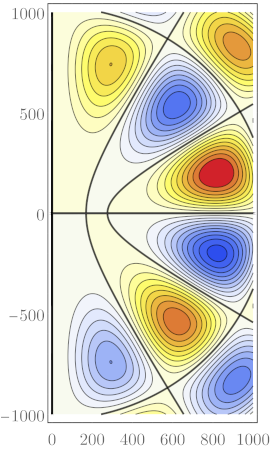}
  \caption{Contour plots of the stream functions $\psi$ obtained using oblate spheroidal coordinates. Results for the first six eigenmodes with $A^2=10$ are shown left to right and top to bottom.}
  \label{fig:case2_A2_10}
\end{figure}

\begin{figure}
  \includegraphics[width=0.32\textwidth]{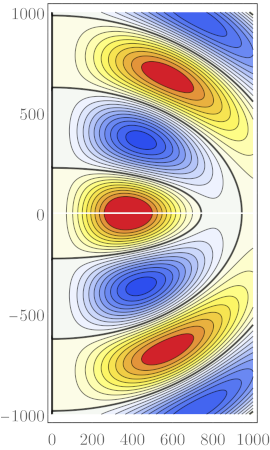}
  \hfill
  \includegraphics[width=0.32\textwidth]{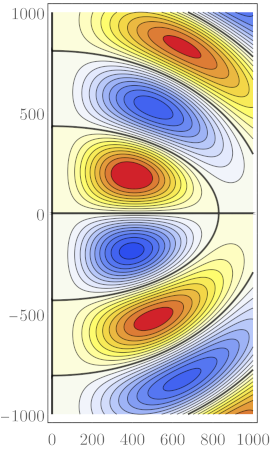}
  \hfill
  \includegraphics[width=0.32\textwidth]{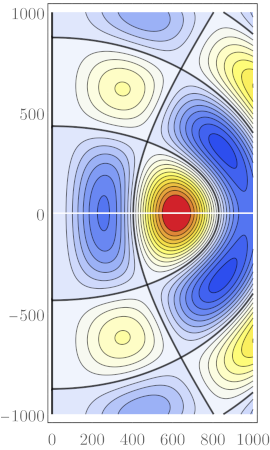}
  \vskip\baselineskip
  \includegraphics[width=0.32\textwidth]{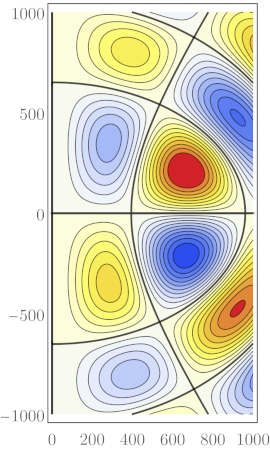}
  \hfill
  \includegraphics[width=0.32\textwidth]{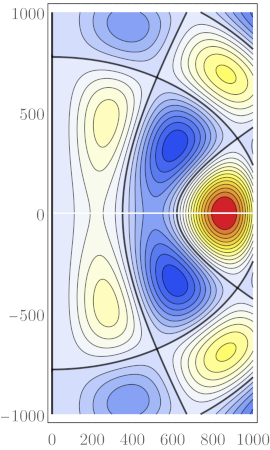}
  \hfill
  \includegraphics[width=0.32\textwidth]{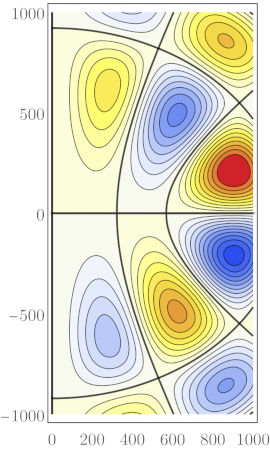}
  \caption{Contour plots of the stream functions $\psi$ obtained using oblate spheroidal coordinates. Results for the first six eigenmodes with $A^2=50$ are shown left to right and top to bottom.}
  \label{fig:case2_A2_50}
\end{figure}

\begin{figure}
  \includegraphics[width=0.32\textwidth]{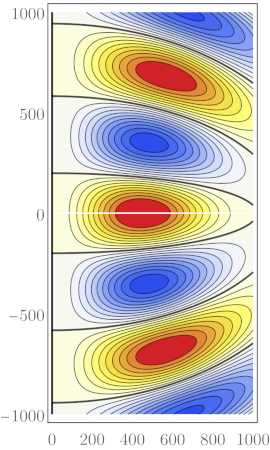}
  \hfill
  \includegraphics[width=0.32\textwidth]{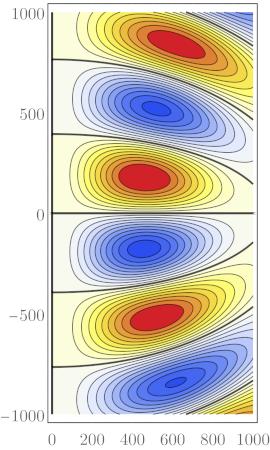}
  \hfill
  \includegraphics[width=0.32\textwidth]{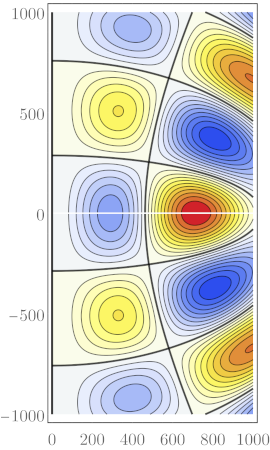}
  \vskip\baselineskip
  \includegraphics[width=0.32\textwidth]{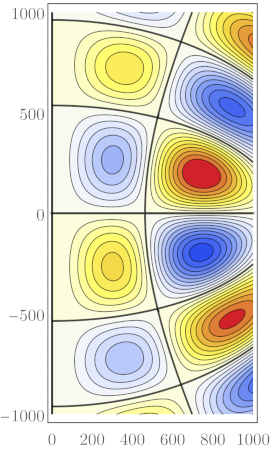}
  \hfill
  \includegraphics[width=0.32\textwidth]{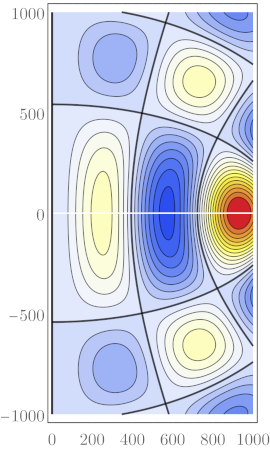}
  \hfill
  \includegraphics[width=0.32\textwidth]{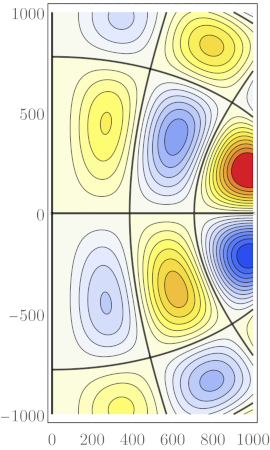}
  \caption{Contour plots of the stream functions $\psi$ obtained using oblate spheroidal coordinates. Results for the first six eigenmodes with $A^2=100$ are shown left to right and top to bottom.}
  \label{fig:case2_A2_100}
\end{figure}

We again note that for large $\mu$ the solution for $F$ in \eqref{eq:FG_case2} resembles that of $F''+A^2F=0$ and therefore exhibits near periodicity in $\mu$. Consequently, the contour plots of the corresponding stream function $\psi$ consist of repeating blocks along some of the curves shown in the right panel of Fig.~\ref{fig:ell-par-cyl} as indicated by the contour plots in Figs.~\ref{fig:case2_A2_1}--\ref{fig:case2_A2_100}.

\section{Vortex Breakdown}
\label{sec:vortex_breakdown}
We noted in Section~\ref{sec:coordinates} that the choice of the coordinate systems used in this paper was made with the intention of modeling tornado-like flows. Specifically, focusing on the corner flow near the origin, we can now address the question whether flows similar to those shown in Fig.~\ref{fig:vortex_breakdown} are possible with Beltrami flows.

Note that the flows that correspond to the even eigenmodes for the spherical and both cases of the spheroidal coordinates have the $r$-axis (or, more accurately the $z=0$ plane) as a stream surface, and therefore the part of the flow where $z>0$ can be taken to model a flow above the (horizontal) ground. We also note that similar flows can be easily created from the odd eigenmodes by applying the third on the list of transformations discussed at the end of Section ~\ref{sec:Bragg} with any $\zeta\in\Real$, similar to what we showed above in Fig.~\ref{fig:psi_paraboloidal} in the context of paraboloidal coordinates.

In Fig.~\ref{fig:vortex_breakdown_flows} we show, from left to right, stream functions that correspond to the fourth eigenmodes of the prolate spheroidal coordinates (left two panels), the spherical coordinates (middle panel), and oblate spheroidal coordinates (right two panels). These can be thought of as snapshots of a continuous transformation in which $a$ in the prolate spheroidal coordinates decreases to $0$, the coordinates becoming spherical coordinates, and then $a$ in the oblate spheroidal coordinates increases from $0$. One can easily visualize this transformation by looking at Fig.~\ref{fig:ell-par-cyl}.
\begin{figure}
  \includegraphics[width=0.19\textwidth]{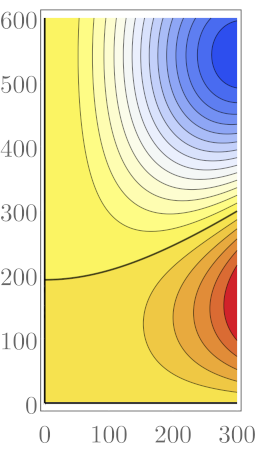}
  \hfill
  \includegraphics[width=0.19\textwidth]{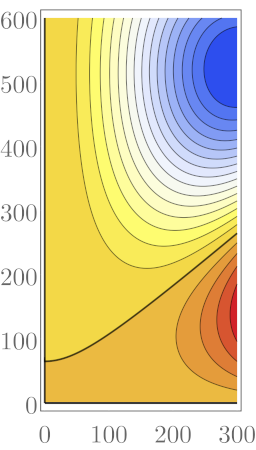}
  \hfill
  \includegraphics[width=0.19\textwidth]{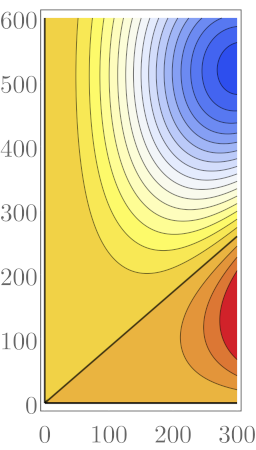}
  \hfill
  \includegraphics[width=0.19\textwidth]{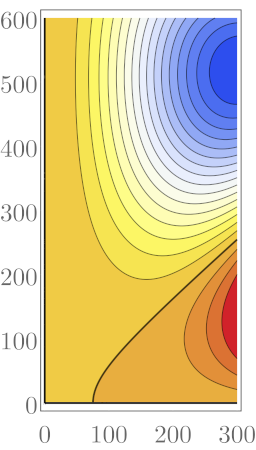}
  \hfill
  \includegraphics[width=0.19\textwidth]{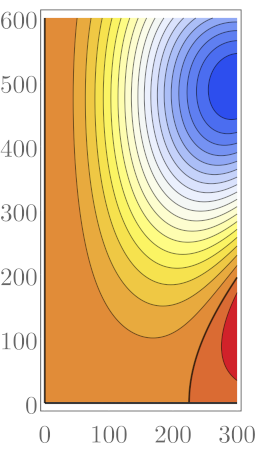}
  \caption{Illustrations of flows that have the characteristics of a vortex breakdown. The middle panel corresponds to spherical coordinates, the left two panels to prolate spheroidal coordinates (with $A^2=10$ and $A^2=1$, respectively), and the right two panels to oblate spheroidal coordinates (with $A^2=1$ and $A^2=10$, respectively). All panels correspond to $\alpha=0.01$ and to the fourth eigenmode.}
  \label{fig:vortex_breakdown_flows}
\end{figure}
One can interpret the left two panels in Fig.~\ref{fig:vortex_breakdown_flows} as a two-cell vortex with a horizontal inflow and a vertical updraft near the corner with a central downdraft near the $z$-axis. The middle panel, the flow in spherical coordinates, corresponds to the flow in which the stagnation point reaches the ground. Finally, the last two panels correspond to the flow in which the central downdraft reaches the ground and results in a horizontal outflow near the ground. Therefore, the progression shown in Fig.~\ref{fig:vortex_breakdown_flows} can be viewed as a quasi-static model for the transition between the flows shown in the middle two panels of Fig~\ref{fig:vortex_breakdown}.

\section{Conclusions}
In fluid dynamics, Beltrami flows have been, among other purposes, used for software validation and hypothesized to occur in tornadic flows. Consequently, it would be beneficial to have a rich catalog of such flows. In this paper we have attempted to construct such flows both analytically and numerically by focusing on Trkalian flows in several orthogonal coordinate systems. Motivated by tornado-like flows, we focused on incompressible, steady, axisymmetric flows which allowed us to use a stream function formulation. After some simplifying assumptions on the flows, we were able to construct solutions to the linear Bragg--Hawthorne equation~\eqref{eq:eqn_psi} in several suitable coordinate systems by reducing the problem to a system of two ordinary differential equations. The obtained solutions have been visualized using contour plots of the stream functions in the $rz$-plane and some of these solutions have been compared to idealized flows thought to occur in two-cell flows with a vortex breakdown. Additionally, we proposed ways to generate infinitely many new stream functions that can be constructed from the existing ones by continuously varying a scalar parameter.

The richness of the solution set obtained in this paper with a constant abnormality $\alpha$ indicates that many other solutions can be found by exploring other three-dimensional coordinate systems and by allowing $\alpha$ to vary in space.

\bibliography{tornado-master}
\bibliographystyle{abbrv}

\end{document}